\documentclass{llncs}

\usepackage[utf8]{inputenc}

\usepackage[T1]{fontenc}
  
\usepackage{enumerate}

\usepackage{graphicx}
\usepackage{subfig}

\usepackage{amsmath}
\usepackage{amssymb}

\title{The hardness of Median in the synchronized bit communication model}

\author{Karolina Sołtys}
\institute{Faculty of Mathematics, Informatics and Mechanics\\[-0.8ex]
\small University of Warsaw \\
\email{ksoltys@students.mimuw.edu.pl}}

\begin{document}
\maketitle
\pagestyle{plain}

\begin{abstract}
The synchronized bit communication model, defined recently by Impagliazzo and Williams in \cite{IW}, is a communication model which allows the participants to share a common clock. The main open problem posed in this paper was the following: does the synchronized bit model allow a logarithmic speed-up for all functions over the standard deterministic model of communication? We resolve this question in the negative by showing that the Median function, whose communication complexity is $O(\log n)$, does not admit polytime synchronized bit protocol with communication complexity $O\left(\log^{1-\varepsilon} n\right)$ for any $\varepsilon > 0$. Our results follow by a new round-communication trade-off for the Median function in the standard model, which easily translates to its hardness in the synchronized bit model.

\end{abstract}

\section{Introduction}
   
Communication complexity, introduced by Yao in 1979 \cite{Yao}, is an important concept in complexity theory which tries to determine the amount of communication needed to compute a function whose input has been distributed among two or more participants. A natural question, especially in the context of the protocols used in distributed computing, is whether \emph{synchronous} protocols, in which the participants can use a common clock, are more powerful than \emph{asynchronous} ones, in which the players do not have this ability (note that the standard protocols studied in communication complexity are asynchronous). The synchronization allows the participants to convey some information by \emph{not} sending a message at a given moment of time, which naturally leads us to examine the time-communication trade-off of the protocol.
 
In their recent paper \cite{IW}, Impagliazzo and Williams formalized the notion of synchronous protocols and partially solved several interesting questions related to them by introducing two models of communication, called the synchronized bit model and the synchronized connection model, and studying their complexity. We can briefly summarize the synchronized bit complexity model as an extension of the standard deterministic model of communication, where a player in one step can send 0, 1 or a blank, and where blanks do not count towards the communication complexity of the protocol. It is interesting to consider the polytime bit complexity of a problem $\Pi$: the minimum complexity of any synchronized bit protocol for the problem $\Pi$ using a polynomial number of steps. This function is denoted by $PB(\Pi)$.

The authors prove in \cite{IW} the following bounds for the polytime bit complexity:
\begin{equation}\label{bit_bounds}
	\Omega\left(\frac{D(\Pi)}{\log n}\right) \leq PB(\Pi) \leq O\left(\frac{D(\Pi)}{\log \log n}\right)
\end{equation}
and conclude their paper with the following questions:
\begin{question}
	Can the upper bound on the polytime bit complexity in (\ref{bit_bounds}) be improved?
\end{question}
\begin{question}
	What is the complexity of Median in this model?
\end{question}

We answer both of these questions by proving the following result:
\begin{theorem}\label{median_hard}
	The Median function does not admit a polytime synchronized bit protocol with communication complexity $O\left(\log^{1-\varepsilon} n\right)$ for any $\varepsilon > 0$.
\end{theorem}
Since the deterministic complexity of the Median function is $O(\log n)$, we get the following lower bound for the synchronized bit complexity of this problem:
$$\omega \left(\frac{D(\textsc{Median})}{\log^{\varepsilon} n}\right) \leq PB(\textsc{Median})$$
for each $\varepsilon > 0$, which of course forbids any significant improvement to the upper bound in (\ref{bit_bounds}) in general.

Theorem \ref{median_hard} can be easily translated, using methods established in \cite{IW}, in terms of the round-communication trade-off in the standard deterministic model.
\begin{theorem}
       The Median function does not admit a deterministic protocol using $O\left(\log^{1-\varepsilon} n\right)$ rounds and a logarithmic amount of communication at each round for any $\varepsilon > 0$.
\end{theorem}
Our result provides a new round-communication trade-off for the Median function. The study of the round-communication trade-offs for various functions is an important area of communication complexity with significant applications to streaming algorithms (lower bounds for the rounds-communication trade-off in the deterministic communication complexity model imply the same bounds for the number of passes-memory trade-off in the streaming model; the opposite implication usually does not hold). With Median being a central problem in the streaming model, the fact that our approach can be used to prove some lower bounds for it in this model (albeit slightly weaker than those already known, \cite{MP}) in a completely different manner, can potentially be quite fruitful. 

To facilitate our proofs we define a natural problem, Strategy, which is easily seen as complete for the class of communication problems solvable in $O(\log n)$ rounds and $O(\log n)$ communication. Our reduction from Strategy to Median allows us to show that Median is also a complete problem for that class.

The structure of the paper is as follows. In section \ref{preliminaries} we briefly describe the relation between synchronized bit complexity and round complexity, and define the problem central to our paper, Strategy. In the following section we prove a round-communication trade-off for the Strategy problem by showing a reduction from the $k(\cdot)$-Pointer-Jumping problem, whose round-communication trade-off has been extensively studied (\cite{NW}, \cite{PRV}). The reduction uses an intermediate problem -- $k(\cdot)$-Level-Strategy. Although both reductions are quite straightforward, they do not preserve the size of the instance, which leads us to a system of asymptotic inequalities on the lower bound for Strategy. We then show a reduction from the Strategy problem to Median.

\section{Preliminaries}\label{preliminaries}
We assume that the reader is familiar with the basic notions of communication complexity \cite{book}.

   We will use the straightforward two-way translation between the synchronized bit protocols running in time $O(t)$ and using $O(b)$ bits of communication and deterministic protocols using $O(b)$ rounds and $O(\log t)$ communication at each round, which is explained in detail in \cite{IW}. We can therefore approach the problem in terms of round complexity: the existence of polytime synchronized bit protocol using $O(b)$ bits of communication is equivalent to the existence of a deterministic protocol using $O(b)$ rounds and $O(\log n)$ communication at each round. 

\begin{definition}
   We define the problem Strategy of size $n$ as follows. Let $T$ be a full binary tree with $n$ vertices. A function $f$ assigns to each vertex of the tree a number from the set $\{0,1\}$. Alice knows the values of $f$ in the vertices in the odd layers of the tree (the vertices with odd depth), and Bob knows the values of $f$ in the vertices in the even layers. We define the \emph{leaf reached by $f$} to be the leaf which is an endpoint of the path starting at the root and going always downwards -- to the left son of $v$ if $f(v) = 0$ and to the right son otherwise. The players' goal is to determine the index of the leaf they reach.
\end{definition}
\begin{remark}\label{strategy_complete}
The Strategy problem is complete for the class of communication problems solvable in $O(\log n)$ rounds and $O(\log n)$ communication. 
\end{remark}
\begin{proof}
We will reduce an arbitrary problem $\Pi$ of this class to the Strategy problem. The problem $\Pi$ has a deterministic protocol using $O(\log n)$ rounds and $O(\log n)$ communication. Without loss of generality we can assume that the protocol makes the players alternate in sending messages containing just one bit. In each Alice's (and analogously Bob's) vertex of the protocol tree and corresponding to a communication history consistent with her input, the message she sends depends only on her input and the communication history (in all the other vertices of Alice we fix her message in an arbitrary way). We can now transform the protocol tree into a Strategy tree by setting the value of the function $f$ in each vertex of the tree to the message sent in that vertex, and assigning to the leaves the outputs of the protocol for the given communication history. \qed
\end{proof}
\begin{remark}\label{strategy_upper}
The Strategy problem can be solved in $O\left(\frac{\log n}{\log \log n}\right)$ rounds and $O(\log n)$ communication at each round. 
\end{remark}
This upper bound may be easily obtained by using the reductions described in \cite{IW} to change the model to the synchronized bit model, use the upper bound proved therein for this model, and then translate the model back to the standard deterministic model.

\begin{definition}
   We define the problem $k(\cdot)$-Level-Strategy of size $n$ as follows. We have an $n$-ary tree $T$ of height $k(n)$, with leaves indexed from $1$ to $n^{k(n)}$. There is a function $f: T \rightarrow [n]$ for each vertex $v \in T$, Alice knows this function for the vertices in the odd layers and Bob knows it for the vertices in the even layers. As in Strategy, they want to determine the index of the leaf they descend to starting from the root and following the function $f$ ($f(v) = l$ means that if they arrive to the vertex $v$ they descend to the $l$-th son of $v$). Note that $n$ is a parameter, and not the input size; both Alice and Bob have input of size $O\left(n^k(n) \log n\right)$.
\end{definition}
\begin{definition}
   The problem $k(\cdot)$-Pointer-Jumping is defined as follows. Alice and Bob each hold a list of $n$ pointers, each pointing to a pointer in the list of the other. An initial pointer $v_0$ is marked. They want to determine the $k(n)$-th pointer they reach after following the pointers starting from $v_0$. 
\end{definition}
In \cite{NW} it was proved that if we allow just $k(n)-1$ rounds then $k(n)$-Pointer-Jumping requires $\Omega(n)$ communication.

\section{Round complexity of Strategy}\label{strategy}
   We will prove the following theorem by showing a sequence of reductions from Pointer Jumping to Strategy:
\begin{theorem}\label{strategy_hard}
   The Strategy function does not admit a deterministic protocol using $O\left(\log^{1-\varepsilon} n\right)$ rounds and a logarithmic amount of communication at each round for any $\varepsilon > 0$.
\end{theorem}

Let $r(n)$ be some function such that $r(\Theta(n)) = \Theta(r(n))$ (we will fix it later). We will prove the following easy lemmas:
\begin{lemma}
   If the Strategy problem of size $m$ can be solved in $O(r(m))$ rounds using  $O(r(m) \log m)$ communication, then, for each $k(\cdot)$, $k(\cdot)$-Level-Strategy can be solved in $O\left(r\left(n^{k(n)}\right)\right)$ rounds using $O\left(r\left(n^{k(n)}\right) \log n^{k(n)}\right)$ communication.
\end{lemma}
\begin{proof}
For each $k(\cdot)$, there is a simple reduction from $k(\cdot)$-Level-Strategy of size $n$ (i. e. with $n^{k(n)}$ leaves) to Strategy of size $O\left(n^{k(n)}\right)$. We will create a Strategy tree $S$ by replacing each Alice's vertex $v$ of the tree $T$ in $k(\cdot)$-Level-Strategy with a binary tree of height $\lceil\log n\rceil$, with $n$ leftmost leaves corresponding to the sons of $v$ in $T$. On every vertex of even depth in the subtree we fix Bob's input to be 0, and we fix Alice's input in the subtree so that the correct leaf is reached. Note that if $n$ is not a power of two, then the Strategy tree $S$ will be slightly larger than $T$, with additional vertices and not corresponding to vertices of $T$, but that is inconsequential, because our construction assures that these vertices are never reached by the protocol. It is easy to see that this is a proper reduction.\qed
\end{proof}

\begin{lemma}
   If $k(\cdot)$-Level-Strategy of size $n$ can be solved in $O\left(r\left(n^{k(n)}\right)\right)$ rounds using $O\left(r\left(n^{k(n)}\right) \log n^{k(n)}\right)$ communication, then $k(\cdot)$-Pointer-Jumping of size $n$ can also be solved in $O\left(r\left(n^{k(n)}\right)\right)$ rounds using $O\left(r\left(n^{k(n)}\right) \log n^{k(n)}\right)$ communication.
\end{lemma}
\begin{proof}
Given an instance $G$ of $k(\cdot)$-Pointer-Jumping of size $n$ we will create an instance of $k(\cdot)$-Level-Strategy of size $n$ (with $n^{k(n)}$ leaves), which will be, informally speaking, a tree of possible paths of length $k(n)$ in the graph $G$. In every odd layer (resp. even layer) every vertex that is an $i$-th son will point its $j$-th son if and only if $g_A(i) = j$, where $g_A$ is the Alice's (resp. Bob's) input in the Pointer-Jumping instance. It is easy to see that in this reduction every player can locally compute their input, and that the vertex reached by their functions in $k(\cdot)$-Level-Strategy is an $i$-th son if and only if the output for the $k(\cdot)$-Pointer-Jumping is $i$. \qed
\end{proof}

By combining the two lemmas we obtain the following
\begin{corollary}
If the Strategy problem of size $m$ can be solved in $O(r(m))$ rounds using  $O(r(m) \log m)$ communication, then, for each $k(\cdot)$, $k(\cdot)$-Pointer-Jumping of size $n$ can be solved in $O\left(r\left(n^{k(n)}\right)\right)$ rounds using $O\left(r\left(n^{k(n)}\right) \log n^{k(n)}\right)$ communication.
\end{corollary}

We can now prove the main theorem of this section.
\begin{proof}[of Theorem \ref{strategy_hard}]
In \cite{NW} it was shown that if we allow no more than $k(n)-1$ rounds then $k(\cdot)$-Pointer-Jumping of size $n$ requires $\Omega(n)$ communication. We thus know that for every $k(\cdot)$ the protocol for Strategy must yield, after using the reductions described, a protocol for $k(\cdot)$-Pointer-Jumping using either a greater number of rounds:
$$r\left(n^{k(n)}\right) \geq k(n)$$
or a greater amount of communication:
$$r\left(n^{k(n)}\right) \log n^{k(n)} \geq \Omega(n).$$

A function $r(n)$ which for any $k(n)$ violates both of these inequalities is thus a viable lower bound for Strategy, that is no $O(r(n))$-round, $O(r(n)\log n)$-communication protocol for Strategy may exist. 

It is easy to check that for all $\varepsilon > 0$ the function $r(n) = \log^{1-\varepsilon} n$ fails to satisfy both of these inequalities when we set $k(n) = \sqrt{\frac{n}{\log n}}$, which proves Theorem \ref{strategy_hard}.\qed
\end{proof}

Note here that if we were to prove a lower bound tightly matching the upper bound for Strategy, that is $O\left(\frac{\log n}{\log \log n}\right)$ (proved in Remark \ref{strategy_upper}), which we believe may be the case, we would need to use a different method, because close examination of the inequalities obtained by setting $r(n)=\frac{\log n}{\log \log n}$ reveals that, regardless of the function $k(n)$, at least one of the inequalities must be satisfied.

\section{Reduction from Strategy to Median}\label{median}
\begin{proposition}
   If the Median problem of size $n$ can be solved in $O(r(n))$ rounds using  $O(c(n))$ communication Strategy of size $n$ can also be solved in $O(r(n))$ rounds using $O(c(n))$ communication.
\end{proposition}
\begin{proof}
We will show here a reduction from Strategy$(T, f)$ to Median$(S, A, B)$, where $S$ is a set of natural numbers and $A, B$ are subsets of this set held by Alice and Bob respectively. The reduction will work inductively on the height $k$ of the tree $T$ of the Strategy problem. 

It is easy to reduce Strategy on trees of height 1 to Median over $S = \{0,1\}$: we give the empty subset to Bob and to Alice either the subset $\{0\}$ or $\{1\}$ depending on the value of $f$ in the root. The two possible values of Median will correspond to the two possible leaves reached by $f$.

Let us denote by $l_i$ the size of the set $S$ produced by the reduction for the trees of size $i$, and by $w_i$ the the number of elements given to Alice and Bob ($w_i = |A \cup B|$); we will construct the reduction inductively so that $l_i$ and $w_i$ are well defined.

We will now show the induction step for the trees of height $k$. Let $T_l$ be the tree of height $k-1$ rooted at the left son of the root, ant $T_r$ be the tree rooted at the right son. We denote by $r$ the root of the tree $T$ and by $A_l$ the subset of $S_{k-1}$ given to Alice by the reduction from Strategy$(T_l, f|_{T_l})$ (we define $A_r, B_l, B_r$ analogously). The reduction will create the sets:
\begin{eqnarray*}
   S & = & S_k = \{1, ..., l_k\} \textrm{, where } l_k = 2w_{k-1} + 2l_{k-1} \\ \\
   A & = & (B_{l} + w_{k-1}) \cup (B_{r} + (w_{k-1} + l_{k-1})) \cup \\
       & & \quad \quad \quad \quad \quad \quad  \cup \quad \begin{cases}
           \{1, ..., w_{k-1}\} &\textrm{if } f(r) = 0 \\
           \{l_{k}-w_{k-1}+1, ..., l_{k}\}    & \textrm{otherwise}
       \end{cases} \\ \\
   B & = & (A_{l} + w_{k-1}) \cup (A_{r} + w_{k-1} + l_{k-1})
\end{eqnarray*}
where $C+d = \{c + d \ |\  c \in C\}$.

It is easy to check that this is a proper reduction. The basic idea is that we give Alice some amount of small numbers if she turns left in $r$, and the same amount of big numbers if she turns to the right, so that the problem reduces either to finding the median in the subproblem of Median corresponding to $T_l$ or to finding the median in the subproblem corresponding to $T_r$, with the roles of the players reversed (because it is now Bob who holds the roots of $T_l$ and $T_r$).

Easy calculations of the recurrence relations for $w_k$ and $l_k$ show that these functions are exponential in $k$, so the reduction from Strategy on the tree of size $n$ produces an instance of Median of size $O(n)$.\qed
\end{proof}

Combining this reduction with Theorem \ref{strategy_hard} yields the proof of our main result, Theorem \ref{median_hard}.

It is worth noticing that this reduction, together with Remark \ref{strategy_complete}, proves as well that Median is complete for the class of communication problems solvable in $O(\log n)$ rounds and $O(\log n)$ communication.

\section{Conclusions}
Our results still hold in the randomized case. 

It is possible that a stronger lower bound for Median in the synchronized bit complexity model, tightly matching the upper bound of $O\left(\frac{\log n}{\log \log n}\right)$, may be proven. It would also be interesting to extend the synchronized bit model to the multiparty case and to study the complexity of the model in this setting.

In \cite{IW} the authors define also another synchronized model: the connection complexity model, which is based on the assumption that in each timestep every party decides whether to try to establish a connection. Some information is exchanged if and only if a connection has been established, and only successful connections count toward the communication cost of the protocol. The model turns out to be surprisingly powerful, enabling the participants to solve the Disjointness problem in polynomial time and only one bit of communication. We believe that also for this model the possible multiparty extension seems worth further examination.  

\medskip {\bf Acknowledgements.} I would like to thank Prof. Iordanis Kerenidis for encouraging me to write this paper and for fruitful discussions.


\begin{thebibliography}{1}

\bibitem{IW}
Russell Impagliazzo and Ryan Williams.
\newblock Communication complexity with synchronized clocks.
\newblock In {\em Proceedings of the 2010 IEEE 25th Annual Conference on
  Computational Complexity}, CCC '10, pages 259--269, Washington, DC, USA,
  2010. IEEE Computer Society.

\bibitem{book}
E.~Kushilevitz and N.~Nisan.
\newblock Communication complexity.
\newblock 1997.

\bibitem{MP}
J.~I. Munro and M.~S. Paterson.
\newblock Selection and sorting with limited storage.
\newblock Technical report, 1978.

\bibitem{NW}
Noam Nisan and Avi Wigderson.
\newblock Rounds in communication complexity revisited.
\newblock {\em SIAM J. Comput.}, 22:211--219, February 1993.

\bibitem{PRV}
Stephen~J. Ponzio, Jaikumar Radhakrishnan, and S.~Venkatesh.
\newblock The communication complexity of pointer chasing.
\newblock {\em J. Comput. Syst. Sci.}, 62:323--355, March 2001.

\bibitem{Yao}
Andrew Chi-Chih Yao.
\newblock Some complexity questions related to distributive
  computing(preliminary report).
\newblock In {\em Proceedings of the eleventh annual ACM symposium on Theory of
  computing}, STOC '79, pages 209--213. ACM, 1979.

\end{thebibliography}
\bibliographystyle{plain}

\end{document}